\documentclass[conference]{IEEEtran}
\usepackage{amsmath}
\usepackage{graphicx}
\usepackage{amsfonts}
\usepackage{amssymb}
\usepackage{multirow}
\usepackage{alltt}
\usepackage{booktabs}
\usepackage{stmaryrd}
\usepackage{dsfont} 
\usepackage{color} 
\usepackage[standard]{ntheorem}
\usepackage[font=footnotesize]{subfig}
\usepackage{array,algorithm2e,algorithmic}
\renewcommand{\theequation}{\arabic{equation}}
\usepackage[width=18.2cm,height=25.7cm]{geometry}
\usepackage{psfig}
\usepackage{slashbox}
\ifCLASSINFOpdf
\else
\fi
\begin{document}

\title{On the design of a family of CI pseudo-random number generators}
\author{\IEEEauthorblockN{Jacques M. Bahi,
Xiaole Fang,
Christophe Guyeux, and
Qianxue Wang}
\IEEEauthorblockA{University of Franche-Comt\'{e}\\
Computer Science Laboratory LIFC,
Besan\c con, France\\ Email:{jacques.bahi, xiaole.fang,
christophe.guyeux, qianxue.wang}@univ-fcomte.fr }}

\maketitle

\begin{abstract}
Chaos and its applications in the field of secure communications have attracted a lot of attention.
Chaos-based pseudo-random number generators are critical to guarantee security over open networks as the Internet.
We have previously demonstrated that it is possible to define such generators with good statistical properties by using a tool called ``chaotic iterations'', which depends on an iteration function.
An approach to find update functions such that the associated generator presents a random-like and chaotic behavior is proposed in this research work.
To do so, we use the vectorial Boolean negation as a prototype and explain how to modify this iteration function without deflating the good properties of the associated generator.
Simulation results and basic security analysis are then presented to evaluate the randomness of this new family of generators.
\end{abstract}



\IEEEpeerreviewmaketitle

\section{Introduction}

Security has become a topic of increasing importance in communications because the Internet and personal communications systems are now accessible worldwide.
To guarantee this security, chaotic systems have many advantages as unpredictability or disorder-like, and they are especially used when complex sequences are required~\cite{Behnia20113455,Hu20092286,DeMicco20083373}.
This is why chaotic systems are frequently used to design new pseudo-random number generators (PRNGs) ~\cite{Behnia20113455,Niansheng}.
Following this approach, we have previously proposed a PRNG based on chaotic iterations. A short overview of our recent researches in this field is given hereafter.

In Ref.~\cite{guyeux09}, it is proven  that chaotic  iterations (CIs),  a suitable
tool for fast computing  iterative algorithms, satisfies the topological chaos
property, as defined by Devaney~\cite{devaney}.  The chaotic behavior of CIs
is used in~\cite{bgw09:ip}, to obtain an unpredictable PRNG that
depends on two  logistic maps.   The resulted  PRNG shows
better   statistical   properties   than   each  individual   component   alone.
Additionally, various chaos properties  have been established. These chaos properties, inherited from CIs, are not
possessed by  the two inputted generators.   We have shown that,  in addition of
being  chaotic,  this generator  can  pass the  NIST  battery  of tests,  widely
considered as a  comprehensive and stringent battery of  tests for cryptographic
applications~\cite{nist}.
Then we have achieved  to improve the speed of the former generator in~\cite{bgw10:ip,guyeuxTaiwan10}, by using ISAAC and XORshift instead of the two logistic maps.
These generators can pass the batteries DieHARD~\cite{diehard} and TestU01~\cite{testU01}.

In these previous researches, the iteration function of CIs was always the vectorial Boolean negation.
We propose now to enlarge the set of iteration functions such that the associated CI-based generator is both chaotic and random-like.
The well-known NIST and DieHARD tests are finally used to evaluate the statistical behavior of this new family of generators.

The rest of this paper is organized as follows. In the next
section, some basic definitions concerning CIs and our PRNG are recalled.
In Section~\ref{section:description} is explained how it is possible to change the iteration function of the generator without losing the good properties of our PRNG.
NIST and DieHARD batteries are passed in Section~\ref{sec:nist} to all of these generators.
 The
paper ends with a conclusion section where our contribution is summarized and intended future work is presented.

\section{Review of Basics}
\label{Basic recalls}

This section is devoted to basic notations and terminologies in the fields of chaotic iterations and PRNGs.

\subsection{Notations}
\begin{tabular}{@{}c@{}@{}l@{}}
$\llbracket 1;\mathsf{N} \rrbracket$ & $\rightarrow\{1,2,\hdots,\mathsf{N}\}$ \\
$S^{n}$ & $\rightarrow$ the $n^{th}$ term of a sequence $S=(S^{1},S^{2},\hdots)$ \\
$v_{i}$ & $\rightarrow$ the $i^{th}$ component of a vector\\
&~~~~$v=(v_{1},v_{2},\hdots, v_n)$\\
$\emph{strategy}$ & $\rightarrow$ a sequence which elements belong in $%
\llbracket 1;\mathsf{N} \rrbracket $ \\
$\mathbb{S}$ & $\rightarrow$ the set of all strategies \\
$\mathds{N}^{\ast }$ & $\rightarrow$ the set of positive integers \{1,2,3,...\} \\
$\mathds{B}$ & $\rightarrow$  $\{0,1\}$ \\
\end{tabular}


\subsection{Chaotic iterations}
\label{subsection:Chaotic iterations}

\begin{definition}
Let $f:\mathds{B}^{\mathsf{N}%
}\longrightarrow \mathds{B}^{\mathsf{N}}$ be an ``iteration'' function and $S\in \mathbb{S}
$. Then, the so-called \emph{chaotic iterations} are defined by~\cite{Robert1986} $x^0\in \mathds{B}^{\mathsf{N}}$ and
$$
\begin{array}{l}
\forall n\in \mathds{N}^{\ast },\forall i\in \llbracket1;\mathsf{N}\rrbracket%
,x_i^n=\left\{
\begin{array}{l}
x_i^{n-1}~~~~~\text{if}~S^n\neq i \\
f(x^{n-1})_{S^n}~\text{if}~S^n=i.\end{array} \right. \end{array}
$$
\end{definition}
In other words, at the $n^{th}$ iteration, only the $S^{n}-$th cell is
\textquotedblleft iterated\textquotedblright.

\subsection{Mapping matrix}

Chaotic iterations introduced above can be described by using the mapping matrix defined bellow.

\begin{definition}
Let $f:\mathds{B}^{\mathsf{N}}\longrightarrow \mathds{B}^{\mathsf{N}}$
be an iteration function, then its associated \emph{mapping matrix}
$\mathsf{f}$ is the matrix of size $\mathsf{N} \times 2^\mathsf{N}$ whose element  $\mathsf{f}_{p,q}$ is the integer having the following binary decomposition:  $q_\mathsf{N}, \hdots, q_{\mathsf{N}-p}, f(q)_{\mathsf{N}-p+1}, q_{\mathsf{N}-p+2}, \hdots, q_1  $, where $q_i$ (resp. $f(q)_i$) is the $i-$th binary digit of $q$ (resp. of $f(q)$).
\end{definition}

The relation between $\mathsf{f}$ and chaotic iterations of $f$ can be
understood as follows. If the current state of the system is $q$ and the
strategy is $p$, then the next state (under the chaotic iterations of $f
$) will be $\mathsf{f}_{p,q}$. Finally, the vector $\mathcal{F}(f)=(f(0),f(1),\ldots,f(2^{\mathsf{N}}-1)) \in \llbracket 0 ; 2^{\mathsf{N}}-1 \rrbracket^{2^{\mathsf{N}}}$ is called \emph{vector of images}.
An example is shown for the vectorial Boolean negation $f_0 (x_1, \hdots, x_\mathsf{N}) = (\overline{x_1}, \hdots, \overline{x_\mathsf{N}})$ in Table~
\ref{negation_output}.
\begin{table*}[!t]
\renewcommand{\arraystretch}{1.2}
\caption{The matrix $\mathsf{f}$ associated to $f_0$}
\label{negation_output}
\centering
\begin{tabular}{c|cccccccccccccccc}
\backslashbox{p}{q} &~~~0 & 1 & 2 & 3 & 4 & 5 & 6 & 7 & 8 & 9 & 10 & 11
& 12 & 13 & 14 & 15~~~ \\ \hline
 $(f_0(q)_1,q_2,q_3,q_4)$ &\multicolumn{1}{r}{\multirow{4}*{$\left(
\begin{array}{c}
8  \\
4 \\
2\\
1\\
\end{array}
\right.$}}&9&10&11&12&13&14&15&0&1&2&3&4&5&6&
\multicolumn{1}{l}{\multirow{4}*{$\left.
\begin{array}{c}
7  \\
11 \\
13\\
14\\
\end{array}
\right)$
}}\\
$(q_1,f_0(q)_2,q_3,q_4)$&&5&6&7&0&1&2&3&12&13&14&15&8&9&10& \\
$(q_1,q_2,f_0(q)_3,q_4)$&&3&0&1&6&7&4&5&10&11&8&9&14&15&12& \\
$(q_1,q_2,q_3,f_0(q)_4)$& &0&3&2&5&4&7&6&9&8&11&10&13&12&15&\\\hline
$\mathcal{F}(f_0)$ &~~~(15,&14,&13,&12,&11,&10,&9,&8,&7,&6,&5,&4,&3,&2,&1,&0)~~~\\
\hline
\end{tabular}
\end{table*}
\subsection{Chaotic iterations as PRNG}
\label{subsec Chaotic iterations as PRNG}

%
%

Algorithm~\ref{Algo:Chaotic iteration} recalls the basic design procedure of our $CI_f(PRNG1,PRNG2)$ generator.
The internal state is $x$, the output array is $r$, and $\mathsf{N} \in \mathds{N}, \mathsf{N} \geqslant 2$. 
Parameters $k$ and $\mathsf{N}$ are constants, PRNG1 picks its values into $\{0;1\}$, and PRNG2 takes a random integer into $\llbracket 1;\mathsf{N} \rrbracket$.
We have previously established that $k$ must be greater than $3\mathsf{N}$ (see \cite{bgw09:ip}).
Finally, until now, $f = f_0$.

\begin{algorithm}
\SetAlgoLined
\KwIn{the internal state $x$ ($\mathsf{N}$ bits)}
\KwOut{a state $r$ of $\mathsf{N}$ bits}
$m\leftarrow{PRNG1()+k}$\;
\For{$i=0,\dots,m$}
{
$S\leftarrow{PRNG2()}$\;
$x_S\leftarrow{f(x)_S}$\;
}
$r\leftarrow{x}$\;
return $r$\;
\medskip
\caption{An arbitrary round of $CI_f(PRNG1,PRNG2)$}
\label{Algo:Chaotic iteration}
\end{algorithm}

This $CI_f(PRNG1,PRNG2)$ generator may utilize any reasonable PRNGs as inputs.
For demonstration purposes, two XORshift are adopted here for both
PRNG1 and PRNG2.
Table \ref{table application example} gives an illustrative example using these PRNGs,
where $\mathsf{N} = k = 4$ and $\mathcal{F}(f)$ = (14, 14,
12, 12, 10, 10, 9, 9, 6, 6, 4, 4, 2, 2, 1, 0).

\begin{table*}
\centering
\begin{tabular}{|c|ccccc|cccccc|cccccc|}
\hline\hline
$m:$ & & & 4 & & & & & 5 & & & & & & 4 & & & \\ \hline
$S$ & 2 & 4 & 2 & 3 & & 4 & 1 & 1 & 4 & 4 & & 3 & 2 & 3 & 3 & & \\
\hline

 & 1& 1& 1& 1&
 & 1& \textbf{1} & \textbf{0} &1 &1 &
 &1 & 1& 1&1 &
 & \\
$f$(x) & \textbf{0} & 1& \textbf{1} & 0 &
 &0 &0& 0&0 & 0&
 &0 &\textbf{0} & 1&1 & &\\
 &1 & 1& 1& \textbf{1}&
 &0 &0 &0 &0 &0 &
& \textbf{0} & 1& \textbf{1} & \textbf{0} &
 & \\
 &0 &\textbf{0} &0 &0 &
 &\textbf{1} &1 &0 &\textbf{1} &\textbf{1} &
 &1 &0 & 0&0 &
 &\\\hline
$x^{0}$ & & & & & $x^{4}$ & & & & & & $x^{9}$ & & & & & $x^{13}$ & \\
4 &0 &0 &4 &6&6 &7 &15 &7 &7 &7 &7&5  &1 &3 &1 &1 &  \\ & & & && & & & &
& &  & & & & & & \\
0 & & & & &
0 & & $\xrightarrow{1} 1$ & $\xrightarrow{1} 0$ & & &
0 & & & & &
0 & \\
1 & $\xrightarrow{2} 0$ & & $\xrightarrow{2} 1$ &  &
1 & & & & & &
1 & & $\xrightarrow{2} 0$ & & & 0 &\\
0 & & & & $\xrightarrow{3} 1$ &
1 & & & & & &
1 & $\xrightarrow{3} 0$ & & $\xrightarrow{3} 1$ & $\xrightarrow{3} 0$ &
0 &\\
0 & & $\xrightarrow{4} 0$ & & &
0 &$\xrightarrow{4} 1$ & & & $\xrightarrow{4} 1$&$\xrightarrow{4} 1$ &
1 & & & & &
1 &\\
\hline\hline
\end{tabular}\\
\vspace{0.5cm}
Binary Output:
$x_1^{0}x_2^{0}x_3^{0}x_4^{0}x_5^{0}x_1^{4}x_2^{4}x_3^{4}x_4^{4}x_5^{4}x_1^{9}x_2^{9}x_3^{9}x_4^{9}x_5^{9}x_1^{13}x_2^{13}... = 0100011001110001...$

Integer Output:
$x^{0},x^{0},x^{4},x^{6},x^{8}... = 6,7,1...$
\caption{Application example}
\label{table application example}
\end{table*}

\section{Description of the selection scheme}
\label{section:description}

In this section is explained how the iteration function $f_0$ can be replaced without losing chaos and randomness.

\subsection{Strong connectivity and chaos}

Let $f:\mathds{B}^\mathsf{N} \rightarrow \mathds{B}^\mathsf{N}$. Its
{\emph{iteration graph}} $\Gamma(f)$ is the directed graph defined as follows. 
The set of vertices is
$\mathds{B}^\mathsf{N}$, and $\forall x\in\mathds{B}^\mathsf{N}, \forall i\in \llbracket1;\mathsf{N}\rrbracket$,
$\Gamma(f)$ contains an arc labeled $i$ from $x = (x_1, \hdots, x_\mathsf{N})$ to $(x_1, \hdots, x_{i-1}, f(x)_i, x_{i+1}, \hdots, x_\mathsf{N})$. 
We have proven in~\cite{GuyeuxThese10} that:

\begin{theorem}
\label{Th:Caractérisation   des   IC   chaotiques}
The $CI_f(PRNG1,PRNG2)$ generator is chaotic according to Devaney if and only if the graph $\Gamma(f)$ is strongly connected.
\end{theorem}

Theorem \ref{Th:Caractérisation   des   IC   chaotiques} only focus on the topological chaos property.
However, it is possible to find chaotic sequences with bad statistical properties, in particular when the iteration function is unbalanced.

\subsection{Obtaining Balanced Maps}
\label{The generation of pseudo-random sequence}


We now explain how to find balanced iterate functions.

\begin{theorem}
Let $j \in \llbracket 1; 2^\mathsf{N} \rrbracket$ and $F = \mathcal{F}(f_0)$ be the (balanced) vectorial Boolean negation: $F_{j}=2^\mathsf{N}-j$.

If $F' = \mathcal{F}(f)$, a vector of images of a \emph{balanced} iterate function $f$, is such that its $j-$th component differs from $F_j$ by only its $i-th$ bit (starting from the right), then $F'_{2^\mathsf{N}-F'_j}=2^\mathsf{N}-j$.
\end{theorem}

\begin{proof}
As $F'_j$ only differs from $F_j$ by its $i-th$ bit, we have: $F'_j=F_j-F_j\&2^{i-1}+(j-1)\&2^{i-1}.$
Therefore, the value $\mathsf{f}'_{i,j}$ of the mapping matrix of $F'$ can be computed as follows:
\begin{equation}
\label{eq2}
\begin{array}{l}
\mathsf{f}'_{i,j}=j_\mathsf{N}j_{\mathsf{N}-1}...f(j)_i...j_1 \\
=(j-1)-(j-1)\&2^{i-1}+F'_j\&2^{i-1}\\
=(j-1)-(j-1)\&2^{i-1}+F_j\&2^{i-1}-F_j\&2^{i-1}+\\
 ~~~~~~~~~~~~~~~~~~~~~~~~~~~~~~~~~~~~~~~~~~~~~~~~(j-1)\&2^{i-1}\\
=(j-1)\\
\end{array}
\end{equation}
The values in $\mathsf{f}$ are uniformly distributed.
However, in the new matrix $\mathsf{N}$, there are twice the value $j-1$ and no $\mathsf{f}_{i,j}$ in the $i$-th row: the uniform distribution is lost. 
To restore the balance, one of the two $j-1$ values must be found and replaced by $\mathsf{f}_{i,j}$. Let $k$ be a variable such that $\mathsf{f}_{i,k}=j-1$ and $\mathsf{f}'_{i,k}=\mathsf{f}_{i,j}$.

As the $i-$th bits in ${f}_{i,k}$ and ${f}'_{i,k}$ are equal, we have:
\begin{equation}
\label{eq5}
\mathsf{f}'_{i,k}\&(2^\mathsf{N}-1-2^{i-1})=\mathsf{f}_{i,j}\&(2^\mathsf{N}-1-2^{i-1}).
\end{equation}

We can thus transform the equation $\mathsf{f}_{i,k}=j-1$ as follows:
\begin{equation}
\label{eq6}
\begin{array}{lll}
\mathsf{f}_{i,k}&=&j-1\\
(k-1)-(k-1)\&2^{i-1}+F_k\&2^{i-1}& =&j-1\\
(k-1)\&(2^\mathsf{N}-1-2^{i-1})+F_k\&2^{i-1}&=&j-1.
\end{array}
\end{equation}

Moreover, from $F_k\&2^i = 2^\mathsf{N}-1-k$, we obtain:
\begin{equation}
\label{eq7}
\begin{array}{lll}
F_k\&2^{i-1}&=&(j-1)\&2^{i-1} \\
(2^\mathsf{N}-k)\&2^i&=&(j-1)\&2^{i-1}\\
(k-1)\&2^{i-1}&=&(k-1)\&2^{i-1}-(j-1)\&2^{i-1}.
\end{array}
\end{equation}

According to Equations (\ref{eq6}) and (\ref{eq7}), we have:
\begin{equation}
\label{eq8}
\begin{array}{ll}
k-1=(j-1)+(k-1)\&2^{i-1}-F_k\&2^{i-1}, \\
\text{where } F_k=2^\mathsf{N}-k \\
 =(j-1)+(k-1)\&2^{i-1}-2^{i-1}+(k-1)\&2^{i-1}\\
 =(j-1)-2^{i-1}+((k-1)\&2^{i-1})*2. \\
\text{But, due to Equation} (\ref{eq8}), \text{ we have:} \\
 =(j-1)-2^{i-1}+((k-1)\&2^{i-1}-(j-1)\&2^{i-1})*2 \\
 =(j-1)+2^{i-1}-((j-1)\&2^{i-1})*2 \\
 =(j-1)+(j-1)\&2^{i-1}+(2^\mathsf{N}-j)\&2^{i-1},\\
\text{where } F_j=2^\mathsf{N}-1-j \\
 =(j-1)+(j-1)\&2^{i-1}+F_j\&2^{i-1} \\
 =\mathsf{f}_{i,j}.
\end{array}
\end{equation}

As
\begin{equation}
\mathsf{f}'_{i,k}=(k-1)-(k-1)\&2^{i-1}+F'_k\&2^{i-1}\\
\end{equation}
and according to Equation (\ref{eq8}), we thus have:
\begin{equation}
\label{eq9}
\begin{array}{ll}
F'_k\&2^{i-1} &=(k-1)\&2^{i-1}.
\end{array}
\end{equation} 

Now, from Equation (\ref{eq5}), we can set that:
\begin{equation}
\label{eq10}
\begin{array}{lr}
\mathsf{f}'_{i,k}\&(2^\mathsf{N}-1-2^{i-1})=\mathsf{f}_{i,j}\&(2^\mathsf{N}-1-2^{i-1}) \\
\multicolumn{2}{l}{((k-1)-(k-1)\&2^{i-1}+F'_k\&2^{i-1})\&(2^\mathsf{N}-1-2^{i-1})=}\\
~~~~((k-1)-(k-1)\&2^{i-1}+F_k\&2^{i-1})\&(2^\mathsf{N}-1-2^{i-1}) \\
F'_k\&(2^\mathsf{N}-1-2^i)=F_k\&(2^\mathsf{N}-1-2^i).\\
\end{array}
\end{equation}

By using both Equations (\ref{eq9}) and (\ref{eq10}), we obtain:
\begin{equation}
\label{eq11}
\begin{array}{ll}
F'_k&=(k-1)\&2^{i-1}+F_k\&(2^\mathsf{N}-1-2^{i-1})\\
&=(k-1)\&2^{i-1}+F_k-F_k\&2^{i-1}\\
&=\mathsf{f}_{i,j} \& 2^{i-1}+(2^N-1-\mathsf{f}_{i,j})-(2^N-1-\mathsf{f}_{i,j}) \& 2^{i-1}\\
&=(2^N-1)-[\mathsf{f}_{i,j}+2^{i-1}-2 \times (\mathsf{f}_{i,j} \& 2^{i-1})] \\
&=(2^N-1)-((j-1)-(j-1)\&2^{i-1}+F_j\&2^{i-1}\\
&+2^{i-1}-2 \times [((j-1)-(j-1) \& 2^{i-1}+F_j \& 2^{i-1})\\
&~~~~~~~~~~~~~~~~~~~~~~~~~~~~~~~~~~~~~~~~~~~~~~~~~~~~~~\& 2^{i-1}\\
&=(2^N-1)-((j-1)-(j-1)\&2^{i-1}-F_j\&2^{i-1}\\
&~~~~~~~~~~~~~~~~~~~~~~~~~~~~~~~~~~~~~~~~~~~~~~~~~~~~~+2^{i-1})\\
&=(2^N-1)-((j-1)+2^{i-1}\&(2^N-j)-F_j\&2^{i-1})\\
&=(2^N-1)-((j-1)+F_j\&2^{i-1}-F_j\&2^{i-1})\\
&=(2^N-1)-(j-1)\\
&=2^N-j.\\
\end{array}
\end{equation}

Finally, from Equation (\ref{eq8}), we can conclude that:
\begin{equation}
\label{eq12}
\begin{array}{ll}
k&=\mathsf{f}_{i,j}+1\\
&=(j-1)-(j-1)\&2^{i-1}+F_j\&2^{i-1} \\
&=2^N-(2^N-j)-(j-1)\&2^{i-1}+F_j\&2^{i-1}\\
&=2^N-(F_j-F_j\&2^i+(j-1)\&2^{i-1})  \\
&=2^N-F'_j.
\end{array}
\end{equation}
\end{proof}

With such equations (namely, Eq. (\ref{eq11}) and (\ref{eq12})), the balance of the new table can be obtained by computing the mapping values. In other words, there is a bijection from the set A of the inputs $x$ into the set B of $F'(x)$ values.

Let us give an example. In Table~\ref{negation_output} is given the mapping matrix for the vectorial Boolean negation, with $\mathsf{N}=4$. Obviously, the values in $\mathsf{f}$ are uniformly distributed: each integer from 0 to 15 occurs once per row. 
Now, if we desire to set $F'_1$ as 14, then $\mathsf{f}'_{4,1}=0$: there will be two $0$ and no $1$ in the fourth row of $\mathsf{f}'$. Due to the previous study, we know that $F'_2$ must be set to 15 too, which leads to $\mathsf{f}'_{4,2}=1$: the balance is recovered.

To sum up, we can determine whether the modification of a bit in the vector of images of the negation function preserves the balance of the outputs or not, by using the following rule (necessary condition):
\begin{itemize}
\item if $F'_j = C$,
\item then $C = F_j-F_j\&2^{i-1}+(j-1)\&2^{i-1}$,
\item and also $F_{2^N-C} = 2^N-j$.
\end{itemize}

This rule, we name it ``Balance Iteration Mapping Rule'', can be used as a criterion to find iterate functions leading to good CI PRNGS, as it is depicted in Algorithm \ref{Chaotic iteration}.
Let us finally remark that, with such a process, it is possible to find new iteration functions by changing more than 1 couple of values in the vectorial Boolean negation $F$. Indeed it is obvious that 2, 3, 4, and even 8 couples of values can be changed using the Balance Iteration Mapping Rule.
For instance, Table~\ref{New vectors of images} contains 8 vectors of images obtained by using Algorithm \ref{Chaotic iteration} one or more times. All of these functions satisfy the hypothesis of Theorem \ref{Th:Caractérisation   des   IC   chaotiques} too, and thus their dynamical systems behave chaotically.


\begin{algorithm}
\SetAlgoLined
\KwIn{a vector of images $F$}
\KwOut{a vector of images $r$ or 0}
\For{$i=0,\dots,2^\mathsf{N}-1$}
{
\For{$j=0,\dots,N$}
{

\If{$F(i+1) \neq F(i+1)-F(i+1) \& 2^j+i \& 2^j$}
{

\If{$F(i+1) \neq 2^N-1-i$}
{
return 0;
}

}

}

\If{$F(2^N-F(i)) \neq 2^N-i$}
{
return 0;
}

}
return $F$\;
\caption{The Balance Iteration Mapping Rule.}
\label{Chaotic iteration}
\end{algorithm}

\begin{table}
\centering
\begin{tabular}{|c|l|}
\hline
Name & Map \\
\hline
$F$&[15,14,13,12,11,10,9,8,7,6,5,4,3,2,1,0]\\
\hline
$F'1$&[14,15,13,12,11,10,9,8,7,6,5,4,3,2,1,0]\\
\hline
$F'2$&[14,15,13,12,9,10,11,8,7,6,5,4,3,2,1,0]\\
\hline
$F'3$&[14,15,9,4,11,8,13,10,7,6,5,12,3,2,1,0]\\
\hline
$F'4$&[14,15,9,12,3,8,13,10,7,6,5,4,11,2,1,0]\\
\hline
$F'5$&[14,15,9,4,11,8,13,10,7,6,5,12,3,2,0,1]\\
\hline
$F'6$&[14,15,9,4,11,8,13,10,3,6,5,12,7,2,0,1]\\
\hline
$F'7$&[14,15,9,4,3,8,13,10,5,2,7,12,11,6,1,0]\\
\hline
$F'8$&[14,15,5,8,9,2,11,12,3,4,13,6,7,10,0,1]\\
\hline
\end{tabular}
\caption{New vectors of images}
\label{New vectors of images}
\end{table}

\begin{table*}[t]
\renewcommand{\arraystretch}{1.3}
\caption{Results through NIST SP 800-22 and DieHARD batteries of tests ($\mathbb{P}_T$ values)}
\label{The passing rate}
\centering
  \begin{tabular}{|l||c|c|c|c|c|c|c|c|}
    \hline
Method & $F'_1$ &  $F'_2$ & $F'_3$ & $F'_4$ & $F'_5$ & $F'_6$ & $F'_7$ & $F'_8$\\ \hline\hline

Frequency (Monobit) Test            &  0.102526 &  0.017912 &  0.171867 &  0.779188 &  0.971699 &  0.275709 &  0.137282 &    0.699313 \\ \hline
Frequency Test within a Block             &  0.085587 &  0.657933 &  0.779188 &  0.897763 &  0.851383 &  0.383827 &  0.262249 &    0.122325 \\ \hline
Cumulative Sums (Cusum) Test*             &  0.264576 &  0.185074 &  0.228927 &  0.736333 &  0.462694 &  0.169816 &  0.391715 &    0.729111\\ \hline
Runs Test                    &  0.739918 &  0.334538 &  0.798139 &  0.834308 &  0.153763 &  0.719747 &  0.534146 &    0.262249 \\ \hline
Test for the Longest Run of Ones in a Block     &  0.678686 &  0.474986 &  0.637119 &  0.037566 &  0.366918 &  0.739918 &  0.236810 &    0.759756 \\ \hline
Binary Matrix Rank Test                & 0.816537 &  0.534146 &  0.249284 &  0.883171 &  0.739918 &  0.037566 &  0.798139 &    0.867692 \\ \hline
Discrete Fourier Transform (Spectral) Test     &   0.798139 &  0.474986 &  0.014550 &  0.366918 &  0.595549 &  0.115387 &  0.798139 &    0.153763 \\ \hline
Non-overlapping Template Matching Test*        &  0.489304 &  0.507177 &  0.477005 &  0.557597 &  0.452278 &  0.505673 &  0.541034 &    0.497140 \\ \hline
Overlapping Template Matching Test        &  0.514124 &  0.171867 &  0.162606 &  0.816537 &  0.319084 &  0.678686 &  0.534146 &    0.798139 \\ \hline
Maurer’s “Universal Statistical” Test         &   0.249284 &  0.171867 &  0.096578 &  0.419021 &  0.171867 &  0.798139 &  0.115387 &    0.275709 \\ \hline
Approximate Entropy Test             & 0.236810 &  0.514124 &  0.262249 &  0.816537 &  0.474986 &  0.080519 &  0.000001 &    0.779188\\ \hline
Random Excursions Test*                &  0.353142 &  0.403219 &  0.229832 &  0.481025 &  0.317506 &  0.602978 &  0.362746 &    0.416274 \\ \hline
Random Excursions Variant Test*            &  0.412987 &  0.369181 &  0.313171 &  0.513679 &  0.274813 &  0.391166 &  0.454157 &    0.341012 \\ \hline
Serial Test* (m=10)                &  0.304324 &  0.102735 &  0.270033 &  0.384058 &  0.456684 &  0.125973 &  0.404429 &    0.253197 \\ \hline
Linear Complexity Test                &0.759756 &  0.153763 &  0.883171 &  0.171867 &  0.366918 &  0.319084 &  0.678686 &    0.075719 \\ \hline
Success                     &15/15&15/15&15/15&15/15&15/15&15/15&15/15&15/15  \\ \hline\hline
Diehard Test               &pass&pass&pass&pass&pass&pass&pass&pass\\ \hline
  \end{tabular}
\end{table*}


\section{Statistical analysis}
\label{sec:nist}
A good random number generator must be indistinguishable from a random sequence through any statistical test. 
As an illustration of the theory presented in this paper, we have used various batteries of tests in order to evaluate the quality of our proposed pseudo random number generator, when iterating functions of Table~\ref{New vectors of images}. These batteries are the well-known and stringent  DIEHARD~\cite{diehard} and NIST~\cite{nist} statistical test suites.

We can conclude from Table \ref{The passing rate} that all of the generators based on the new iterate functions have successfully passed both the NIST and DieHARD batteries of tests. 
These results show the good statistical properties of the proposed PRNGs, and thus the interest of the theoretical approach presented in this paper.

\section{Conclusion and future work}
\label{Conclusions and Future Work}

In previous researches, we have presented a pseudo-random number generator based on chaotic iterations.
It depends on an iteration function, formerly fixed to the negation function.
We have previously established a characterization of functions leading to a chaotic behavior for the associated generator.
However, this characterization allows unbalanced functions, whose generator cannot pass statistical tests.
We have proposed in this paper an algorithm that can find iteration functions leading to a chaotic generator statistically irreproachable. 
This algorithm has been used to find 8 functions such that their generators are both chaotic and able to pass the NIST and DIEHARD statistical batteries of tests.

In future work, we will continue to explore conditions that improve the randomness of the associated CI PRNGs. New statistical tests will be used to compare these PRNGs to existing ones, and a cryptanalysis of our generator will be proposed. 
Finally, new applications in computer science will be proposed, especially in the Internet security field.

\bibliographystyle{plain}
\bibliography{mabase}

\end{document}